\documentclass[twoside,leqno,twocolumn]{article}  
\usepackage{ltexpprt} 
\usepackage{amssymb,latexsym}
\usepackage{amsmath}
\usepackage{graphicx}
\usepackage{tikz}
\usepackage{fancybox}
\usepackage{caption}
\usepackage{subcaption}

\def\cal{\mathcal}
\def\E{\mathbb E}

\def \hat {\widehat}

\def \und {\underline}
\def \m{{\boldsymbol m}}
\def \u {{\boldsymbol u}}

\def\EOP{{\vrule height 5pt width 5 pt depth 0 pt}}

 \newcommand{\comments}[1]{ \begin{center} {\fbox{\begin{minipage}[h]{0.9
            \linewidth} {\sf #1}  \end{minipage} }} \end{center}}

\title{The recurrence function of a random  Sturmian word.}

\author{Pablo {\sc Rotondo} and Brigitte {\sc Vall\'ee}  } 

\date{July 10, 2016} 

\begin{document}

\maketitle  

\abstract{}  This paper  describes the probabilistic behaviour of a  random Sturmian word.  It performs the  probabilistic  analysis of the recurrence function which  can be viewed as a waiting time to discover all the factors of length $n$ of the Sturmian word. This parameter is  central to combinatorics of words. Having fixed a possible length $n$ for the factors,  we let $\alpha$ to be drawn uniformly from  the unit interval $[0,1]$, thus defining a random Sturmian word of slope $\alpha$. Thus the waiting time for these factors becomes a random variable, for which we study the limit distribution and the limit density.

\section{Introduction}  

\subsection*{Recurrence and Sturmian words.}

The recurrence function   measures  the  ``complexity'' of an infinite word and  describes the possible occurrences of finite factors inside it together with the maximal  gaps between successive occurrences.  This recurrence  function  is   thus widely studied, notably  in the case of Sturmian words.  Sturmian words are central in combinatorics of words, as they   are  in a  precise sense the simplest  infinite words which are not eventually periodic \cite{Lot}. With each Sturmian word is associated  an irrational number $\alpha$, which is called the slope of the Sturmian word, and  many  of its characteristics 
 depend on the continued fraction expansion   of $\alpha$.  This is in particular the case for the recurrence  function $n \mapsto R(\alpha, n)$, where  
 the  integer   $R(\alpha, n)$  is the length of the smallest ``window''  which is  needed for discovering  the  set  ${\cal L}_\alpha (n)$   of  all the finite factors of length $n$  inside  $\alpha$. As this set ${\cal L}_\alpha (n)$ is widely used in many applications of Sturmian words (for instance   
quasicrystals,  or  digital geometry),   the  function $n \mapsto R(\alpha, n)$ thus  intervenes very often  as a pre-computation cost.

 \medskip
   From  a result due to  Morse and Hedlund \cite{Morse},   it is known that  the recurrence function  $R(\alpha, n)$ 
   depends on  $\alpha$ via  its continued fraction expansion, and, notably,  its  {\em continuants}. The continuant $q_k(\alpha)$ is the denominator of the $k$-th convergent of $\alpha$, and, for an irrational $\alpha$,  the sequence $k \mapsto q_k(\alpha)$ is strictly increasing.  
The result of  Morse and Hedlund   
expresses $R(\alpha, n)$ in terms  of  the integer $n$ together with the two ends   of the interval $[q_{k-1}(\alpha), q_k(\alpha)[$ which contains $n$. More precisely,   for any  $ n \in 
 [q_{k-1}(\alpha),  q_k(\alpha)[ $, one has 
 \begin{equation} \label{R} R(\alpha, n)=n-1+q_k(\alpha)+q_{k-1}(\alpha) \, .
 \end{equation} 

  It is thus natural to study the recurrence function via 
  \begin{equation} \label {quot}
  S(\alpha, n):= \frac {R(\alpha, n) +1}{ n} = 1 + \frac{q_{k-1}(\alpha)}{n}+\frac {q_k(\alpha)}{n} \, , 
  \end{equation}
  called the recurrence quotient.
   Most of the classical studies   deal  with a {\em fixed} $\alpha$, and   the  usual focus is put on {\em extremal}  behaviours  of the recurrence function.  The following result exhibits a large variability of the function $n \mapsto S(\alpha, n)$.

  \begin{proposition}\label{Morse}  The following holds  for the recurrence quotient defined in \eqref{quot}: 
\begin{itemize}
\item[$(i)$]   
For  any  irrational  real $\alpha$, 
  one has
  $$ \liminf_{n\to\infty} S(\alpha, n) \le 3.$$
 \item [$(ii)$] {\rm [Morse and Hedlund]  \cite{Morse}}
  For almost any   irrational $\alpha$,  and any $ c>0$ one has
$$ \limsup _{n\to\infty}\frac {S(\alpha, n)}{ \log n}= +\infty,\quad 
\limsup_{n\to\infty} \frac {S(\alpha, n)}{ (\log n)^{1+ c}}= 0 $$
\end{itemize}
\end{proposition}
  This result also shows that the quotient recurrence  is ``small''  for integers $n$ which are close to the  right end of the interval $[q_{k-1}(\alpha), q_k(\alpha)[$, whereas  it is ``large'' when $n$ is close to $q_{k-1}(\alpha)$ (see Figure 1).

 \begin{figure*} \label{recquot}
\centering
\begin{subfigure}[b]{0.48\textwidth}
\includegraphics[width = 6cm]{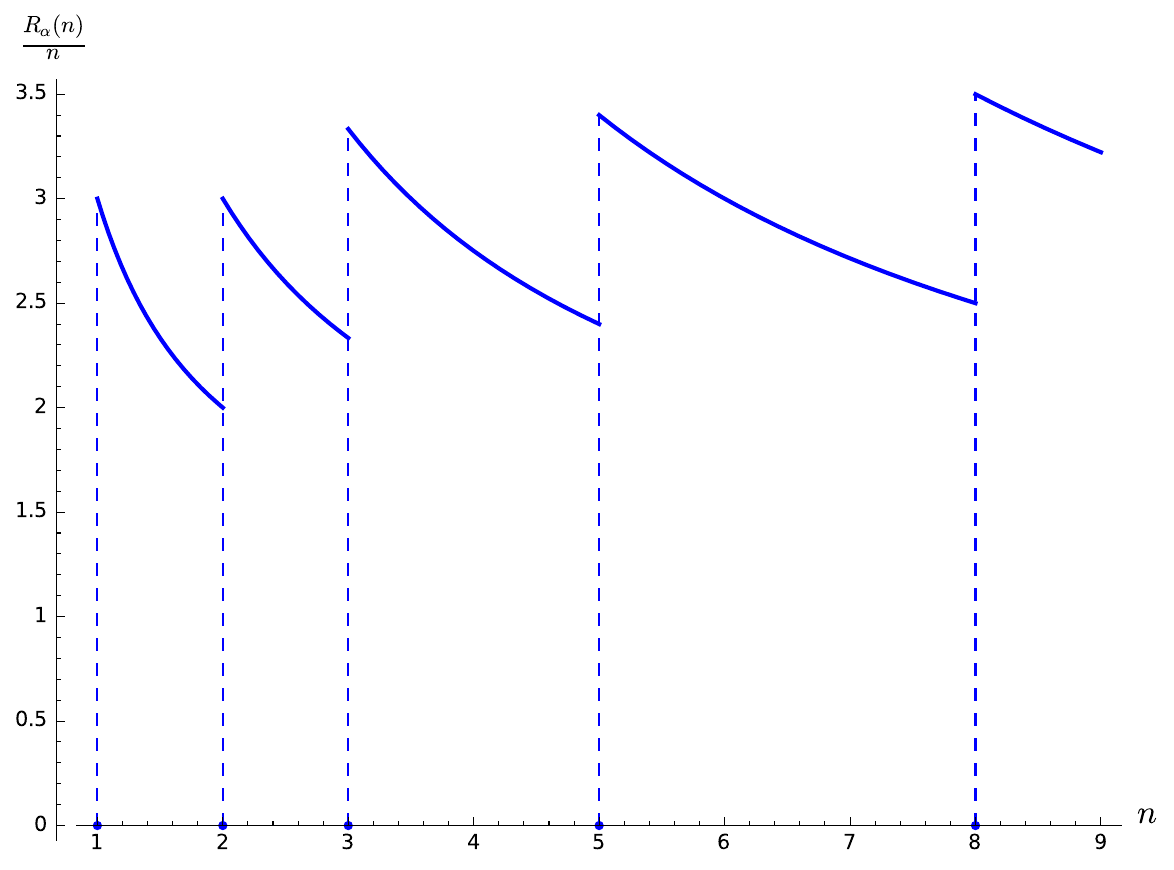}
\caption{ Recurrence quotient for $\alpha=\varphi^2$, \\with $\varphi = (\sqrt{5} - 1)/{2}$.}
\end{subfigure}
\begin{subfigure}[b]{0.48\textwidth}
\includegraphics[width=6cm]{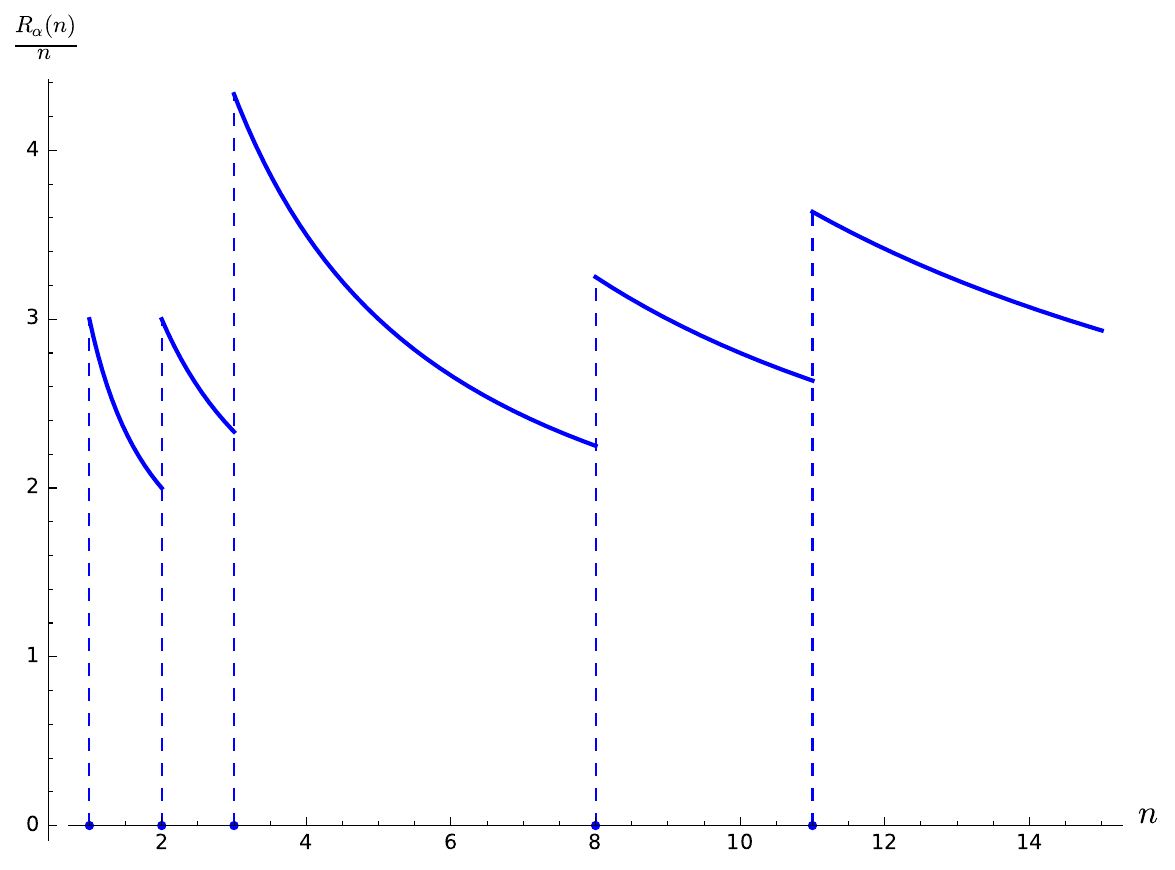}
\caption{Recurrence quotient for $\alpha=\varphi^2$, \\with $\alpha = 1/e$.}
\end{subfigure}
\caption {Examples of recurrence quotients for  two Sturmian words.}
\end{figure*}

 \subsection*{Two  different  probabilistic settings.}
  Here, we adopt a probablistic approach,   and   consider a random Sturmian word, associated with a random irrational slope $\alpha$ of the unit interval.  There are now two possibilities: 
  
  \begin{itemize}
  \item [$(i)$] fix the integer $n$ (corresponding to the length of the factors, which will further tend to $\infty$); now the index $k$  of the interval $[q_{k-1}(\alpha), q_k(\alpha)[$ which contains $n$ is a random variable $k= k(\alpha, n)$. 
   This model may be  called  the model ``with  a large fixed $n$''.
The sequence $n \mapsto S(\alpha, n)$ is  now a  sequence of random variables. 

\item[$(ii)$] fix a depth $k$  (that further tends to $\infty$),  and a fixed $\mu \in [0, 1]$.  For  any slope $\alpha$, we consider the interval $[q_{k-1}(\alpha), q_k(\alpha)[$ delimited by  the  two successive continuants with indices $k-1$ and $k$, and we choose there the integer $n:= n_\mu(\alpha, k)$ at a barycentric position, which is now a random variable.   This model may be  called  the model ``with  a large fixed $k$''. The  sequence $k \mapsto S_k^{\langle \mu \rangle} (\alpha) :=S(\alpha, n_\mu(\alpha, k))$ is  now a   sequence of random variables.

\end{itemize}
In each case, we are interested  by the same type of questions  about the sequence of random variables: Does there exist a limit   for the expectations? a limit distribution? a limit density?

 
  
   \subsection*{The main results.}  We have already performed the  probabilistic study  of type $(ii)$ (the model with ``a large fixed $k$'') in \cite{mfcs}, and we return to it in Section 5.1. \\
   We here deal  with the recurrence quotient within the model $(i)$  (the model with ``a large fixed $n$'').  We obtain three results for the recurrence quotient; we consider  the random variables $\alpha \mapsto  S(\alpha, n)$ and study them for large $n$. We exhibit a limit  for their distribution,  and prove that there exists a limit density, as $n \to \infty$. We also study the  conditional expectation of the recurrence quotient, when we exclude the  possibility for $n$ to be too close of the  left end of the interval  $[q_{k-1}(\alpha), q_k(\alpha)[$.  And we exhibit  a class of events, for which the order of this conditional mean value  is  exactly of order $\log n$. This can be viewed as a probabilistic extension of the Morse and Hedlund result (compare with Proposition \ref{Morse}).  
  
  \smallskip 
   Our proofs use elementary methods: they are based on a precise comparison between  an integral and its Riemann sum ; however, the integral is  improper (but convergent)  and the Riemann sum  is constrained by a coprimality condition, what we call a ``coprime Riemann sum''. 
   
   \smallskip
   We also introduce a general family of functions, called  continuant-functions or ${\cal Q}$-functions, which  are defined via the  sequence of continuants $k \mapsto q_k(\alpha)$. The recurrence quotient is an instance of such a function,  but the other  ``geometric'' parameters of interest provide other natural examples of such a notion.  And the  framework of the paper is well-adapted to the  study of  a general ${\cal Q}$ function.

 \subsection*{Plan of the paper.}  Section 2   gives a precise definition of the parameters under study, introduces the class of ${\cal Q}$-functions and states our three results: limit distributions in Theorem \ref{thm1}, limit densities in Theorem \ref{thm2},   and conditional expectations in Theorem \ref{thm3}. Section 3 is devoted to the proof  of the first two results, whereas Section 4 focuses on the  study of conditional expectations.  Section 5 compares the results obtained in the two models, the present model (with large fixed $n$), and the  model  (with large fixed $k$)  previously  studied in  \cite{mfcs} .

\section{General framework and main results.}
The section first makes precise the  notions that were informally defined in the introduction, notably Sturmian words and 
recurrence. Then, it introduces  parameters which describe the  geometry of the ``continuant intervals'' or  the  position of the integer $n$ inside the continuant interval.  Section 2.3 defines the  class  of  ${\cal Q}$ functions that provides a convenient framework for  our study.   Then,  we state Theorems \ref{thm1}  \ref{thm2}  in Sections 2.5 and 2.6, for general ${\cal Q}$-functions. We return to our specific parameters of interest, notably the recurrence function in Section  2.7, with  two figures (Figures 2 and 3).  Section 2.8 concludes with conditional expectations.

\subsection{More on Sturmian words and  recurrence function.}
  We consider a   finite set  ${\cal A}$ of {\em symbols}, called   {\em alphabet}. Let $\u =(u_n)_{n \in \mathbb {N}}$ be an infinite word  in  ${\mathcal A }^{\mathbb N} $.
 A  finite word $w$  of length $n$ 
is a factor of    $\u$ 
 if there exists an index $m$ 
for which $w= u_{m}\dots u_{m+n-1}$.  
Let  ${\cal L} _{ \u }(n)$ stand for the  set  of   factors of length $n$  of $\u$.  
Two functions describe the set ${\cal L} _{ \u }(n)$ inside the word $\u$, namely  the complexity and the recurrence function.  
 
 \smallskip
 {\bf \em Complexity.}
 The (\emph{factor}) \emph{complexity function} of the infinite word $\u$  is defined as the sequence $n \mapsto p _{ u }(n):= |{\cal L} _{ \u }(n)|$.  The eventually periodic  words are the simplest  ones, in terms
of  the complexity function, and satisfy  $ p _{ \u }(n) \le n$ for  some $n$. 

\smallskip
 The simplest words  that  are not eventually periodic  satisfy the equality  $p _{ \u }(n) = n + 1 $ for each $n\ge 0.$  Such  words do exist, they are called   {\em Sturmian words}.  
Moreover,  Morse and Hedlund  provided  a powerful arithmetic description of Sturmian words (see also \cite{Lot} for more on Sturmian words).
   
 \begin{proposition} {\rm [Morse and Hedlund] \cite{Morse}} \label{prop:sturm} {Associate with  a pair 
 $(\alpha, \beta)\in [0, 1]^2$  
 the two  infinite words $\underline {\frak S}(\alpha, \beta)$ and $\overline {\frak S}(\alpha, \beta)$ whose $n$-th  symbols  are respectively  $$
\underline  u_n = \lfloor \alpha(n+1) + \beta\rfloor -  \lfloor \alpha n  + \beta\rfloor,$$
$$   \overline u_n = \lceil \alpha(n+1) + \beta\rceil -  \lceil \alpha n  + \beta\rceil.$$ Then a word $\u\in \{0, 1\}^{\mathbb N}$ is Sturmian if and only if  it  equals  $\underline {\frak S}(\alpha, \beta)$ or  $\overline {\frak S}(\alpha, \beta)$ for   a pair $(\alpha, \beta)$   formed  with an irrational $\alpha \in ]0, 1[$   and a real  $\beta\in [0, 1[$.}
\end{proposition}

\smallskip 
{\bf \em Recurrence.}
It is also important  to study  where   finite factors occur inside the infinite word $u$.
  An infinite word  $\u \in {\mathcal A}^{\mathbb N}$ is   \emph{uniformly recurrent}  if every factor of $\u$ 
appears infinitely often and with bounded gaps. More precisely,   denote by $w_\u(q, n)$ the  minimal number of symbols $u_k$ with $k \ge q$ which have to be inspected for discovering  the whole set ${\cal L}_\u (n)$ from the index $q$. Then, the integer $w_\u(q, n)$ is a sort of   ``waiting time'' and   $\u$ is uniformly recurrent if  each set $\{w_\u(q, n) \mid q \in {\mathbb N}\}$  is bounded, and      the {\em recurrence function} $n \mapsto   R _{ \u}(n)$   is defined as 

\vskip0.1cm
\centerline{
$R _{ \u}(n) :=  \max  \{w_u(q, n) \mid  q \in {\mathbb N} \}$.}
\vskip 0.1cm 
 We then recover  the usual definition: Any  factor  of length  $R _{\u}(n)$  of  $u$ contains all the factors of length $n$  of $\u$, and the length $R _{\u }(n)$ is the smallest integer which satisfies this
 property. \\ 
  The inequality $R _{\u }(n) \ge p _{ \u }(n) + n-1$ thus  holds.

\smallskip
Any Sturmian word  is uniformly recurrent. Its 
  recurrence function   only depends on  the slope $\alpha$ and  is thus denoted by $ n \mapsto R(\alpha, n)$.  As we already said, it only depends on  $\alpha$ via  the sequence of its  {\em continuants} $k \mapsto q_k(\alpha)$, and satisfies \eqref{R}.

\subsection{Position parameters.}
  Besides the recurrence quotient, there are also three other parameters $\nu, \mu, \rho$ which describe the geometry of the interval $[q_{k-1}(\alpha), q_k(\alpha)[$ which contains $n$ (this is the case for $\rho$)  or the position of $n$ inside this interval (the case for $\mu$ and $\nu$)
  \begin{equation} \label {rho}
     \rho (\alpha, n) = \frac{q_{k-1}(\alpha)}{q_k(\alpha)}\, ,  
     \end{equation}
     \vskip -0.5cm 
  \begin{equation} \label{munu}
   \mu(\alpha, n): = \frac{n-q_{k-1}(\alpha)} {q_k(\alpha)-q_{k-1}(\alpha)}, \quad  \nu(\alpha, n) =  \frac {n}{q_k(\alpha)} \, . 
   \end{equation}
  
  When $n$ belongs to the interval  $ [q_{k-1}(\alpha),  q_k(\alpha)[$,   the  recurrence quotient    is expressed with  $\rho$ and $\nu$ as  
  \begin{equation} \label{ratio}
  S(\alpha, n)  
 =  1 +\frac{1 +\rho(\alpha, n)}{\nu(\alpha, n)} . 
 \end{equation} As $\nu(\alpha, n)$ belongs to the interval $[\rho(\alpha, n), 1]$, the following bounds hold
  \begin{equation} \label{ratio1}
  2 + \rho(\alpha, n)  \le S(\alpha, n) 
  \le   2 + \frac {1}{\rho(\alpha, n)}
  \end{equation}
 (the lower bound  holds for $n$  close to $q_k(\alpha)$  whereas  the  upper bound is attained  for $n= q_{k-1} (\alpha)$).
 
 \smallskip
   The ratio $\rho
   (\alpha, n)$ belongs to $]0, 1]$, and 
 the Borel-Bernstein   Theorem  proves
 that $\liminf _{n\to\infty}  \rho(\alpha, n)= 0$  for almost any   irrational $\alpha$.  This is the main step for proving  Proposition \ref{Morse}.

  \begin{figure*}
\centering
\begin{center}
{\def\arraystretch{2.5}\tabcolsep=15pt
  \begin{tabular}{| c | c || c |}
    \hline
    Parameter & Function $f(x,y)$ & Density $\frac{12}{\pi^2}J_f(\lambda)$ \\ \hline\hline
    $S$ & $1+x+y$ & $ \begin{cases}	\frac{12}{\pi^2}\frac{1}{\lambda-1}\log(\lambda-1) &\mbox{if } 2\leq \lambda\leq 3 \\
 \frac{12}{\pi^2}\frac{1}{\lambda-1}\log(1 + \frac{1}{\lambda-2}) & \mbox{if } \lambda\geq 3\,. \end{cases}$ \\ \hline
    $\rho$ & $\displaystyle \frac{x}{y}$ & $\displaystyle \frac{12}{\pi^2}\frac{1}{1+\lambda}|\log\lambda|$  \qquad  for  $0\leq \lambda\leq 1$ \\ \hline
    $\mu$ & $\displaystyle \frac{1-x}{y-x}$ & $\displaystyle \begin{cases} 	\frac{12}{\pi^2}\frac{1}{2\lambda -1}\left(2\log 2 - \frac{\log \lambda}{\lambda-1}\right) &\mbox{if } \lambda \neq 1/2 \\
\frac{24}{\pi^2}\left(1-\log 2\right) & \mbox{if } \lambda = 1/2\,. \end{cases}$ \\ \hline
    $\nu$ & $\displaystyle \frac{1}{y}$ & $\displaystyle \frac{12}{\pi^2} \frac 1 \lambda \log(1+\lambda)$  \qquad  for  $0\leq \lambda\leq 1$ \\
    \hline
  \end{tabular}
  }
\end{center}
		\label{fig:limiting densities}
        \caption{Limit densities for the main parameters.}
\end{figure*}

\subsection{${\cal Q}$-functions.}
More generally, we  are interested in functions whose definition strongly depends  on the  partition defined by the continuants, and consider  the  functions  $(\alpha, n) \mapsto \Lambda (\alpha, n)$  that are associated with  some  function $f$ and  are written as,  
\begin{equation} \label{Lambda} 
\Lambda(\alpha, n) =  f\left(\frac {q_{k-1} (\alpha) }{n},  \frac {q_{k} (\alpha) }{n}\right), \end{equation}
 as soon as $n \in [q_{k-1} (\alpha), q_k (\alpha)[$.\\ 
 In the following,    we restrict ourselves to  a  function   $f$   that satisfes the following three properties
  \begin{itemize} 
 \item[$(i)$]   it is  written as  the non trivial quotient of two linear functions
 \begin{equation} \label{coeff}  f(x, y) = \frac{a_1 x+b_1 y+c_1}{ a_2x +b_2 y + c_2} \, ;
 \end{equation}
 \item[$(ii)$]  it is  defined on the  unbounded rectangle 
$$\cal{R} := \{ (x, y) \mid 0 < x \le 1 < y \}, $$ 
  \item[$(iii)$]  it is   non negative on ${\cal R}$ .
  \end{itemize} 
A function  $\Lambda$ which is written  as in \eqref{Lambda} in terms of such a function $f$  is called a continuant-function, or a ${\cal Q}$-function. 

\medskip
Our four parameters of interest, namely the recurrence quotient,  the ratio $\rho$ and the two  parameters which  describe the position of integer $n$ with respect to the interval $[q_{k-1}(\alpha), q_k(\alpha)[$ are ${\cal Q}$-functions, associated to the following functions $f $
$$ f_S (x, y) = 1 + x+ y, $$
$$\ \  f_\rho (x, y) = \frac x y , \quad  f_\mu(x, y)= \frac {1-x}{y-x}, \quad f_\nu(x, y) = \frac 1  y\, .$$

\subsection{Probabilistic setting.} 
We recall the present setting, already described in the introduction.
 We consider a fixed integer $n$, and  a random real $\alpha$ in the unit interval $[0, 1]$. The sequence $\Lambda_n(\alpha):= \Lambda(\alpha, n)$ is  now a  sequence of random variables. We are interested in the  limit distribution of the sequence when $n \to \infty$.   Does there exist a limit distribution? a limit density?

 \subsection{General results - distributions.}
 In the distribution study, we  associate  with a real $\lambda\ge 0$ the subdomain of ${\cal R}$,  
\begin{equation}   \label{Deltaf}
\Delta_f(\lambda):= \{(x, y) \mid 0 \le x \le1 \le y ; \ f(x, y) \le \lambda\}\, 
\end{equation} 
(which is  a convex domain due to the particular form of the function $f$), 
and associate the integral 
\begin{equation} \label {If}
 I_f (\lambda) =  \iint_{\Delta_f(\lambda) }\omega(x, y) dx dy = I[\omega, \Delta_f(\lambda)]  \, , 
 \end{equation}
which involves the    function  $\omega$ defined on ${\cal R}$  as 
\begin{equation}  \label{omega} \omega(x, y)  = \frac {1}{y(x+y)} \, ,  
\end{equation}
whose integral on ${\cal R}$ satisfies $I(\omega, {\cal R}) = \pi^2/12$. The associated  density 
\begin{equation} 
\label{psi}
\psi (x, y)=  \frac {12}{\pi^2}  \frac {1}{y(x+y)}
\end{equation}
plays a fundamental role in the sequel, as  
our originally discrete distribution smooths out (converges weakly) to  the  distribution associated with the density $\psi$, as the following result shows:

\begin{theorem}
\label{thm1}
Consider a ${\cal Q}$-function associated with a function $f$. 
Then 
the sequence $n \mapsto \Lambda_n(\alpha)$ as $n\to\infty$  admits a limit distribution, and  the sequence 
\begin{equation}
F_n( \lambda) := {\mathbb{P}}\left[\Lambda_n \leq \lambda\right] = \frac{12}{\pi^2}   I_f(\lambda)   + O\left(\frac{1}{n}\right)\,,
\end{equation}
involves the integral $I_f (\lambda)$ defined in  \eqref{If}. Moreover, 
 the constant  does not  depend on  the pair $(f, \lambda)$  
\end{theorem}

\begin{figure*}
\centering
\begin{subfigure}[b]{0.48\textwidth}
                \includegraphics[width=\textwidth]{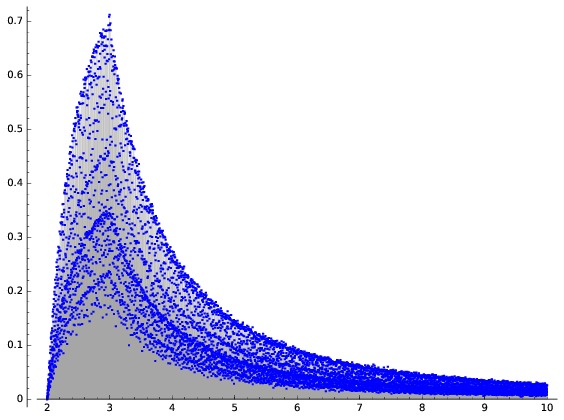}
                \caption{An experimental histogram for $S(\alpha,n)$ with step $\epsilon(n)=1/n$.}
                \label{fig:limdensitymun1}
        \end{subfigure}
\quad
\begin{subfigure}[b]{0.48\textwidth}
                \includegraphics[width=\textwidth]{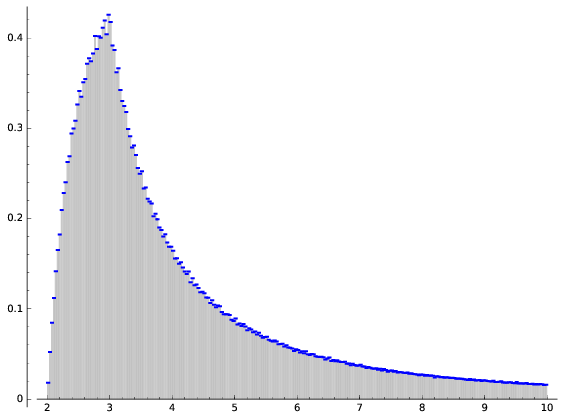}
                \caption{An experimental histogram for $S(\alpha,n)$ with step $\epsilon(n)=1/\lceil\sqrt{n}\rceil$.}
                \label{fig:limdensitys1}
        \end{subfigure}
        \caption{Limiting densities for  the sequence $n \mapsto S(\alpha, n)$ as estimated by the scaled histograms. The number of experiments is $M=10^7$, while $n=1000$. The histograms have been scaled so that they integrate to $1$.} 
\end{figure*}

\subsection{General results - densities.}
For the densities, we deal 
with boundary curves
$  \{(x, y) \mid 
f(x, y) =  \lambda\} $ and their intersection with ${\cal R}$.
We prove the following:

\begin{theorem} \label{thm2}
 Consider a ${\cal Q}$-function associated with a function $f$  which is written as  in \eqref{coeff}.  Then, 

\smallskip 
 $(a)$   The function  $\lambda\mapsto  I_f(\lambda)$ and its derivative $J_f$  
exist for any $\lambda$. The derivative   $J'_f $  exists except  perhaps on a finite set,  which contains the point $b_1/b_2$  and two possible other values $\lambda_0$ and $\lambda_1$.  The following holds: 

\begin{itemize}

\item[$(i)$] At  each of the points $ \lambda = \lambda_i$,  the function $J_f$ admits a left and a right derivative, each of them being finite. 

\item[$(ii)$] When the determinant $r(a, b) := a_1 b_2 -a_2 b_1$ is zero, the derivative $J'_f$ exists at $\lambda = b_1/b_2$. \
\item[$(iii)$] When the determinant $r(a, b) := a_1 b_2 -a_2 b_1$ is not zero, the derivative  $J'_f$ does not exist at $b_1/b_2$ and is $O(|b_2 \lambda- b_1|^{-1})$  for $\lambda \to b_1/b_2$.
\end{itemize}

\smallskip 
 $(b)$  
 For any  stricly positive sequence $n \mapsto \epsilon (n)$ which tends  to 0   with $n \epsilon(n)  \to \infty$,  the  secants of the distribution $F_n$    with  step $\epsilon(n)$ converge to  $
J_f(\lambda)$ and the following holds  
\begin{equation}
 \frac{F_n(\lambda+\epsilon(n))-F_n(\lambda)}{\epsilon(n)} = \frac{12}{\pi^2}  J_f(\lambda) +  E ( \lambda, \epsilon(n)) \, , 
\end{equation} 

$(c)$ The error term satisfies
$$
 E(\lambda, \epsilon(n) ) = O\left(\frac{1}{\epsilon(n) n}\right)  + O\Big( | J_f'(\lambda)| {\epsilon(n)}\Big)\,,
$$
 and  the constants in the $O$-term  do not depend on the pair  $(f, \lambda)$. 
 \end{theorem}

\subsection {Return to the parameters under study.}
We now apply the  previous two results to the quotient recurrence $S$ together with the three geometric parameters.  Figure 2 exhibits  the  limit densities, whereas Figure 3  focuses on the recurrence quotient  and compares scaled experimental histograms   to  the limit density. 

\smallskip
Our theorems entail  the following  estimates for the recurrence quotient 
 $$ \lim_{n \to \infty} {\mathbb P}  \left[  S_n \in [2, 3]\right] = \frac{6}{\pi^2} (\log 2)^2\, ,$$
  $$\hbox{for $ b \ge 2$} \quad  \lim_{n \to \infty} {\mathbb P} \left[  S_n \ge   b+1\right]  = \frac {12}{\pi^2}{\rm Li}_2\left( \frac 1 b\right)\, , $$
  and involves the dilogarithm 
  $  \displaystyle {\rm Li}_2 (x) := \sum_{k \ge 1} \frac {x^k}{k^2}\, .$

\subsection{Conditional expectations.}
We now focus on the  position parameters  $\rho$, $\nu$ and $\mu$ defined in \eqref{rho} and  \eqref{munu}, and consider the  three sequences
$$\rho_n (\alpha) := \rho(\alpha, n), \ \ \nu_n(\alpha)  :=\nu(\alpha, n), \ \  \mu_n(\alpha)  :=\mu(\alpha, n)\, . $$ We have explained that the largest values of the recurrence quotient arise when $\nu$ or $\mu$ are small.  In particular, the event $[\nu_n \ge \epsilon(n)]$ gathers the reals $\alpha$ for which the integer $n$ is not too close of the left end of the interval $[q_{k-1} (\alpha), q_k(\alpha)[$,  and, at the same time, the length of the interval $[q_{k-1} (\alpha), q_k(\alpha)[$ is of the  same order  as  the right end $q_k(\alpha)$.
We then consider a sequence $\epsilon(n) \to 0$, and  condition with  one of the events 
$$ [\rho_n \ge \epsilon(n)], \ \  [\nu_n \ge \epsilon (n)], \ \  [\mu_n \ge \epsilon(n)]\, .$$
  
\begin{theorem} \label{thm3}
Consider   a parameter $\Gamma \in \{ \rho, \mu, \nu\} $ defined in \eqref{rho} and \eqref{munu}.  
Then the conditional expectation  
of the recurrence quotient  $S_n$  with respect to  the event $[\Gamma_n \ge \epsilon(n)]$ 
satisfies 
$$ 
{\mathbb{E}}\left[  S_n \Big| \Gamma_n\ge   \frac 1 n \right]   =  \frac {12}{\pi^2} \log n + O(1) \, .
$$ 
\end{theorem}

This result  exhibits a sequence of events,  for which the integer is not too close to the left ends of  interval $[q_{k-1}(\alpha), q_k(\alpha)[$.   When we are sure not to be too close to this left end, the  quotient of the recurrence quotient is  (on average) of order $\log n$. This can be viewed as a counterpart of Proposition \ref{Morse}. We return to this study  at the end of Section 4. 

\section{Proofs of Theorems \ref{thm1} and \ref{thm2}.} \ \ 

We first  introduce the main objects of interest: continued fraction expansions and coprime Riemann sums.  Then, we prove  the existence of  limit distribution and limit densities for a general ${\cal Q}$ function.  The proof of Theorem \ref{thm1} has three main steps, described in Sections  3.3, 3.4, and 3.5, and we conclude the proof of Theorem \ref{thm1} in Section 3.6.  Section 3.6 
is devoted to the proof of Theorem \ref{thm2}. 

\subsection{Continued fractions, fundamental  intervals and continuants.} 
(See here \cite {hardy} for more details). The continued fraction of  an irrational number $\alpha$ of the unit interval $[0, 1]$ is 
 $$
\alpha = \cfrac{1}{m_1+\cfrac{1}{m_{2} + \cfrac{1}{\ddots +  \cfrac 1 {m_k+  \cfrac {1}{\ddots}} }}} \, .$$
Truncated at depth $k$, it gives rise to a rational number $p_k/q_k$  associated with  a coprime  integer pair $(p_k, q_k)$. The  numerator $p_k = p_k(\alpha)$ and the denominator $q_k = q_k(\alpha)$ are uniquely defined by the irrational number $\alpha$. All the irrational numbers $\alpha$ which begin  with the same sequence $\m  = (m_1, m_2, \ldots, m_k) \in {\mathbb N}^k$ belong to  an interval, called a fundamental interval of depth $k$ and  denoted  by $I_k(\m)$.   As the irrational numbers   of  $I_k(\m)$ have the same convergents of order $\ell \le k$, we denote their numerator and denominator by $p_\ell(\m), q_\ell(\m)$.  The ends of the  interval $I_k(\m)$  are
$$ \frac {p_k(\m)} {q_k(\m)}, \ \  \frac {p_k(\m)+ p_{k-1} (\m) } {q_k(\m)+q_{k-1}(\m)}.$$
As the equality $|p_k(\m) q_{k-1}(\m) -  p_{k-1}(\m) q_{k}(\m)| = 1 $ holds, the length of the fundamental interval  involves the function $\omega$ defined in \eqref{omega} under the form
\begin{equation} \label{fund-int}
 |I_k(\m)| = \omega (q_{k-1}(\m), q_k(\m)) \, .
 \end{equation}
 This explains why the   function $\omega$  defined in \eqref{omega}  and the associated density $\psi$ are  ubiquitous in the study of  the ${\cal Q}$-functions.

\subsection{Distributions. Strategy of the proof.}
There are  two main steps in the proofs   of Theorem \ref{thm1}.

\begin{itemize}
\item[$(i)$]  {\em Discrete step.} We express in Proposition \ref{prop:discreteprob} the  distribution of   a ${\cal Q}$ function
in terms of a  variant of a Riemann sum, that is called in the following a ``coprime'' Riemann sum.  This type of  ``constrained'' Riemann sum  was already  considered in \cite{boca}.
\item[$(ii)$] {\em Continuous step.} We compare the  ``coprime'' Riemann sum to the associated integral. We begin  by the comparison of the ``plain'' Riemann sum  to the integral in Proposition  \ref{sumintegral}, then,  we take into account the coprimality condition in Proposition \ref{coprimesum}.  We extend here the results of \cite{boca} which are only proven  for finite domains. 
\end{itemize}

\subsection{Distributions and Riemann sums.}
We begin with the alternative  expression  of a ${\cal Q}$-function $\Lambda$, associated with $f$,  (already defined in \eqref{Lambda}),  which is  written with  the  Iverson bracket\footnote{The Iverson bracket is  a Boolean function defined by $[\!\![ {\cal P}]\!\!] = 1$ as soon as Property ${\cal P}$ is true} under the form
$$\Lambda(\alpha, n) $$
$$=  \sum_{k \ge 0} f\left(\frac {q_{k-1} (\alpha) }{n},  \frac {q_{k} (\alpha) }{n}\right)  \Big[\!\!\Big[ n \in [q_{k-1} (\alpha), q_k(\alpha)[\Big]\!\!\Big] \, .
$$
The distribution of a ${\cal Q}$-function associated with $f$ is 
$${\mathbb P} [\Lambda_n \le  \lambda] = \int_0^1  d\alpha \sum_{k \ge 0} \left[\!\!\left[ \left(\frac {q_{k-1} (\alpha)} {n} , \frac {q_k(\alpha) } {n}\right)  \in  \Delta_f(\lambda)\right]\!\!\right] \, .$$
For each $k$, the  family of  fundamental intervals $I_k(\m)$  defines a pseudo-partition when $\m$  goes through   ${\mathbb N}^k$,  
and,  for any $\alpha  \in I_k(\m)$,  the equality $q_k (\alpha)= q_k(\m)$ holds.    We deduce
$${\mathbb P} [\Lambda_n \le  \lambda] $$
$$ =  \sum_{k=0}^\infty \sum_{\m \in {\mathbb N}^k}  |I_k(\m)|  \left[\!\!\left[ \left(\frac {q_{k-1} (\m)} {n} , \frac {q_k(\m) } {n}\right)  \in  \Delta_f(\lambda)\right]\!\!\right]\, . $$
Then,  with the expression of the length $|I_k(\m)|$ in terms of the function $\omega$ given in \eqref{fund-int} and the fact that $\omega$ is homogeneous of degree -2, we   obtain
$$ |I_k(\m)| =  \frac 1 {n^2}\   \omega\left( \frac {q_{k-1} (\m)} {n} , \frac {q_{k} (\m)} {n} \right).$$
Now, as we go through all the sequences $\m \in {\mathbb N}^\star $, the  coprime pairs  $(q_{k-1}(\m), {q_k}(\m))$  give rise to  all the coprime pairs $(a, b)$. Moreover, each  coprime pair  $(a, b)$, except  the pair $(1, 1)$,  appears exactly twice,  due to the   existence of two  continued fraction expansions,   the proper one (in which the last digits strictly greater than $1$), and the improper one (in which the last digit is equal to $1$). Then,  the equality holds
$${\mathbb P} [\Lambda_n \le  \lambda] 
  =  \frac 2 {n^2} \sum_{(a,b) \in  \mathbb{Z}^2\atop{(a, b) = 1}} \omega\left(\frac{a}{n},\frac{b}{n}\right)  \left[\!\!\left[ \left(\frac  a  {n} , \frac {b} {n}\right)  \in  \Delta_f(\lambda)\right]\!\!\right]\, .$$
  The right member is  the Riemann sum
  of the function $2\omega$ on the domain $\Delta_f(\lambda)$ with step $1/n$, 
  with an  extra condition of coprimality. 
More generally, for a function $g$ integrable on a subset $\Omega$, we  are led  to the following  two Riemann sums with step $1/n$:   the first one  $R_n(g, \Omega)$ is the usual one, 
\begin{equation*}
 R_n\left( g,  \Omega\right)  = \frac 1 {n^2} \sum_{(a,b) \in  \mathbb{Z}^2} g\left(\frac{a}{n},\frac{b}{n}\right)  \left[\!\!\left[ \left(\frac  a  {n} , \frac {b} {n}\right)  \in \Omega\right]\!\!\right],
\end{equation*}
 whereas   the second one  $\hat R_n(g, \Omega)$ takes into account the coprimality of $(a, b)$, and is called the ``coprime'' Riemann sum, \begin{equation*}
\hat R_n\left( g,  \Omega\right)  = \frac 1 {n^2} \sum_{(a,b) \in  \mathbb{Z}^2\atop{ \gcd(a, b) = 1}} g\left(\frac{a}{n},\frac{b}{n}\right)  \left[\!\!\left[ \left(\frac  a  {n} , \frac {b} {n}\right)  \in \Omega\right]\!\!\right],
\end{equation*}
We summarize:

\begin{proposition}
\label{prop:discreteprob}
Consider a ${\cal Q}$-function $\Lambda$ associated with a function $f$. 
Then  the  distribution  $F_n( \lambda) := {\mathbb{P}}\left[\Lambda_n \leq \lambda\right] $ is  expressed with  a coprime Riemann sum, 
\begin{equation}
{\mathbb P} [\Lambda_n \le  \lambda]  =   \hat R_n\left( 2 \omega,  \Delta_f(\lambda)\right)\, .
\end{equation}
which involves the density $\omega$ defined in \eqref{omega} and the domain $\Delta_f(\lambda)$ defined in  \eqref{Deltaf}.
\end{proposition}

The previous result extends  if we  replace $\Delta_f(\lambda)$ by any other domain $\Omega\subset {\cal R}$.  In particular, in Section  4, we  will  deal with  two ${\cal Q}$-functions $\Lambda$ and $\Gamma$ associated respectively to $f$ and $g$,  together with  the domain 
\begin{equation}  \label {Deltafg}
\und \Delta_{f,g}(\lambda,\epsilon):=\{(x,y)\in \mathcal{R}\mid  f(x,y) \geq \lambda, g(x,y) \geq \epsilon\} \, , 
\end{equation}
and use 
 the equality 
\begin{equation} \label{LambdaGamma} 
{\mathbb{P}}[\Lambda_n\geq \lambda, \Gamma_n \geq \epsilon] = \hat R_n\left( 2 \omega,  \und \Delta_{f,g}(\lambda,\epsilon)\right)\, . 
 \end{equation}

\subsection{Usual Riemann sums and integrals.}
We first deal with the usual Riemann sum, and compare it to its associated integral $I(g, \Omega)$. This is a classical proof, but we consider improper integrals and we wish to have precise error terms.

 We  now deal (only within this subsection) with  
 \begin {equation} 
 \label {S}{\cal S} :=   [0,1]\times (0,\infty) \, ,  
 \end{equation}
consider a subset $\Omega \subset {\cal S}$
and associate  with  it the   family of subsets 
$$\Omega(k) := \Omega \cap \left([0, 1]\times [k,  k+1] \right), $$ 
 for any $k \ge 1$, which form a pseudo-partition of $\Omega$.  
 We also consider  a  positive function $g$ defined on 
$\Omega$ of class ${\cal C}^1$, bounded on any bounded subset on $\Omega$ 
for which the  following   two finite  bounds\footnote{By convention,  we  consider that $C_g(\Omega,k)$ and $D_g(\Omega,k)$ are $0$ if the set $\Omega(k)$ is empty.}
    $$ C_g(\Omega, k):= \sup \{  g(x, y) \mid (x, y) \in \Omega(k) \} , $$
$$ D_g(\Omega, k) := \sup\left\{  \left| \frac {\partial g}{\partial y} (x, y) \right|  \mid (x, y)  \in \Omega(k) \right\}\, , $$
define 
sequences  
whose associated series are convergent, 
$$ \sum_{k \ge 0} C_g(\Omega, k) < \infty, \qquad\sum_{k \ge 0} D_g(\Omega, k) < \infty\, . $$
Their sums are denoted by $C_g(\Omega)$ and $D_g(\Omega)$, and  we denote by $M_g(\Omega)$ their maximum. 

\medskip
Such a function $g$ is called strongly decreasing on $\Omega$  with bound $M_g(\Omega)$. Such a function  is integrable on $\Omega$ and the inequality $I(g, \Omega) \le M_g(\Omega)$ holds.

\begin{proposition} 
\label{sumintegral}
 Consider  the domain ${\cal S}$ defined in \eqref{S} and  a function $g$ which is strongly decreasing on a convex $\Omega\subset{\cal S}$ with bound $M_g(\Omega)$.  
 Then,  the Riemann  sum of the function $g$ on $\Omega$   compares to the integral, 
\begin{equation}
 \left|R_n (g , \Omega) - I (g, \Omega) \right|   \le    \frac 5 n  M_g(\Omega)\, . \end{equation}
\end{proposition}

\begin{proof}
 We will prove   the estimate, for each $k\ge 0$, 
 $$\left| R_n (g , \Omega(k)) -I(g, \Omega(k)\right|  \le \frac 4 n \left( C_g(\Omega, k)+ D_g (\Omega, k)\right).$$
 This will entail the result by taking the sum over  $k \ge 0$. 

\medskip
\noindent  We consider the  elementary squares  of side $1/n$, namely $$\mathcal{R}_{a,b}=\left[\frac{a}{n},\frac{a+1}{n}\right] \times \left[\frac{b}{n},\frac{b+1}{n}\right], $$ and we concentrate on  those which meet $\Omega(k)$.  There are two cases for  such  rectangles 
${\cal R}_{a, b}$, namely
$$ (i) \ \ \mathcal{R}_{a,b}\subset \Omega(k), \qquad \hbox{or}\quad (ii) \ \   \mathcal{R}_{a,b}\cap \Omega(k)^c \not  = \emptyset \, .$$
In the first case $(i)$,  the definition of the bound $D_g$ entails the estimate 
$$\left|\frac{1}{n^2} g\left(\frac{a}{n},\frac{b}{n}\right) - I(g , {\cal R}_{a, b})\right| $$
$$\leq  I\left(\left| g\left(\frac{a}{n},\frac{b}{n}\right)  - g \right|, {\cal R}_{a, b}\right) \le    \frac 1 {n^3} \, D_g(\Omega, k)\, .$$
As the number of such squares is at most $n^2$, 
the contribution from  case  $(i)$ 
is at most $(1/n)\,  D_g(\Omega, k)$.

\smallskip
In the second case $(ii)$,  the positivity of $g$ and the definition of the bound $C_g$ entails the estimate 
\begin{align*}
\left|\frac{1}{n^2} g\left(\frac{a}{n},\frac{b}{n}\right) -  I (g, \Omega\cap \mathcal{R}_{a,b} )\right| \leq &  \frac 1 {n^2} \, C_g(\Omega, k)\, .
\end{align*}
But, the convexity of $\Omega$  entails that there are at most $4 n$  such squares,  and the contribution of the  second case  is 
at most $(4/n)\,  C_g(\Omega,  k)$.  The constant 4 is explained in the Annex. \EOP

\end{proof}

\subsection{Coprime Riemann sums and integrals.}  The following result is an extension of the results obtained in \cite{boca}, that are only proven for finite domains.

\begin{proposition}  \label{coprimesum}
 Consider a  positive function $g$  defined on   ${\cal R}$,  homogeneous of degree $-\beta$ there with  $\beta>1$. Such a function is strictly decreasing  on ${\cal R}$. Consider also  a convex subset  $\Omega\subset \cal R$.  Then,  the  coprime Riemann  sum of the function $g$ on $\Omega$  compares to the integral of $g$ on $\Omega$, namely
$$
 \left|\hat R_n (g , \Omega) - \frac{6}{\pi^2} I(g, \Omega) \right|
   \le    \frac 1 n \left(  1 + 5\zeta( \beta) \right)  M_g({\cal R}) \, .
   $$
 \end{proposition}

\begin{proof}
To filter the cases in which $\gcd(a,b)>1$,  we  use the Mobius function $\mu$  which performs ``inclusion-exclusion''. The Mobius function  $\mu: {\mathbb N} \rightarrow \{-1, 0, +1\}$ satisfies 
\begin{equation}
\label{eq:mucondition}
\sum_{d|n} \mu(d) = \begin{cases} 	1 &\mbox{if } n=1 \\
0 & \mbox{if } n>1\,. \end{cases}
\end{equation}
We  consider the restricted  ``coprime'' Riemann sum, where  the sum is taken over the  pairs $(a, b)$ with $\gcd(a, b) = 1$, namely
$$
n^2 \, \hat { R}_n (g, \Omega)  =
 \sum_{(a,b) \in  \mathbb{Z}^2\atop{\gcd(a, b) = 1}} g\left(\frac{a}{n},\frac{b}{n}\right)  \left[\!\!\left[ \left(\frac  a  {n} , \frac {b} {n}\right)  \in \Omega \right]\!\!\right]\, .$$ 
We then ``insert'' the $\mu$-function  inside this restricted ``coprime'' Riemann sum, 
$$ n^2 \hat { R}_n (g, \Omega)  $$
$$  =  \sum_{(a,b)\in  \mathbb{Z}^2} g\left(\frac{a}{n},\frac{b}{n}\right)   \left[\!\!\left[ \left(\frac  a  {n} , \frac {b} {n}\right)  \in \Omega\right]\!\!\right]  \left(\sum_{d|\gcd(a,b)} \mu(d) \right).$$
As the point $(a/n, b/n)$ belongs to ${\cal R}$ with $a>0$,   the inequality $\gcd(a, b) \le n$ holds. Then, interverting the summations entails the equality  
$$  n^2 \, \hat {R}_n (g, \Omega) $$
$$=  \sum_{d\leq n} \mu(d)\sum_{(a,b )\in \mathbb{Z}^2} g\left(\frac{a d}{n},\frac{b d}{n}\right)  \left[\!\!\left[ \left(\frac  {ad}  {n} , \frac {bd} {n}\right)  \in \Omega\right]\!\!\right]. $$ Finally, the following equality holds
\begin{equation}  \hat {R}_n (g, \Omega)  = \sum_{d\leq n} \mu(d)\,  R_n(g_d, \Omega_d) \, ,   
\end{equation}
and involves the function $g_d$ and the subset $\Omega_d$ defined as
$$ g_d(x, y) := g(dx, dy), \qquad   \Omega_d= \frac  1 d {\Omega} \, .$$

\smallskip 
As  the inclusion  $\Omega_d\subset{ \cal S}$ holds, we now apply the previous Proposition \ref{sumintegral} to each  (plain) Riemann sum  $R_n(g_d, \Omega_d)$ and  obtain
\begin{equation} \label{deb}
 \left|R_n (g_d , \Omega_d) - I (g_d, \Omega_d) \right|   \le    \frac 5 n  M_{g_d}(\Omega_d) \, .\end{equation}
 We  now use three properties.  We  first remark   the equality
$$  I(g_d,  \Omega_d)  = \frac 1 {d^2}  I(g, \Omega)\, , $$
due to the change of variables $(x', y') = (dx, dy)$.  Second,  the series of general term $\mu(d) /d^2$ is  convergent, and, with the Mobius inversion,  its  sum  equal $1/\zeta(2)$ and 
$$ \left| \sum_{d \le  n} \frac {\mu(d)}{d^2} - \frac{ 6} {\pi^2 }\right|   \le \frac 1 n\, .$$ 
  Third,  we relate  the bound  $ M_{g_d} (\Omega_d)$   to  its analogous.    As $g$ is homogeneous of degree $-\beta $,  its derivative is homogeneous of degree $(-\beta -1)$ and  the   two relations  $$
  g_d(x, y) = g(dx, dy) = \frac 1 {d^\beta} g(x, y) \, ,$$
  $$ \frac  {\partial g_d} {\partial y}(x, y) = d  \frac {\partial g} {\partial y}(dx, dy) = 
   \frac 1 {d^\beta}   \frac  {\partial g} {\partial y} (x, y)\, , $$ 
  hold for $(x, y) \in {\cal R}$.   As $g$ and its derivative  are 0 outside ${\cal R}$, the same  holds for $g_d$ and its derivative, and
  $$ M_{g_d}( \Omega_d) = M_ {g_d}( \Omega_d\cap {\cal R}) =  \frac 1 {d^\beta}  M_g  (\Omega_d \cap{\cal R})\le 
   \frac 1 {d^\beta}  M_g  ({\cal R}) \, .$$
   Then, as $\beta>1$, one has
   $$ \sum_{d \le n}  M_{g_d}( \Omega_d) \le \zeta (\beta)  M_g  ({\cal R}) \, .$$

  \smallskip
With the three  previous properties, together with   Eq. \eqref{deb}, 
we obtain  the final result.
\end{proof}

 \subsection {Distributions. Proof of Theorem \ref{thm1}.} 
   Theorem \ref{thm1} is a particular case of  the previous Proposition  \ref{coprimesum}, when it applies to $\omega$ and $\Delta_f(\lambda)$  defined in \eqref{omega}  and \eqref{Deltaf}. The function $\omega$ is  homogeneous of degree -2   and the domain  $\Delta_f(\lambda)$ is convex, as it is the intersection of the unbounded rectangle ${\cal R}$ with the halfplane $\{  f(x, y) \le \lambda\} $. 
   Applying Proposition \ref{coprimesum}  
  then  proves Theorem \ref{thm1}.

 \subsection {Proof of  Theorem \ref{thm2}.}
  Assertion $(a)$   is proven in the Annex. We prove now Assertion $(b)$. 
 We let 
$$F_n(\lambda):= {\mathbb P} [\Lambda_n \le \lambda], \qquad  F_\infty (\lambda)= \frac {12}{\pi^2} I_f(\lambda) \, .$$
 We know that  the derivative   $J_f(\lambda)$ of  $\lambda \mapsto I_f(\lambda)$  exists.  This is the same for the function $F_\infty$ and  we wish  to estimate the  difference 
$$\left|\frac{F_n(\lambda+\epsilon(n))-F_n(\lambda)}{\epsilon(n)}-F_\infty^\prime(\lambda)\right| \, .$$

We begin with the triangle inequality  \begin{equation} \label {ineq}
\left|\frac{F_n(\lambda+\epsilon(n))-F_n(\lambda)}{\epsilon(n)}-F_\infty^\prime(\lambda)\right|
 \end{equation}
$$
 \leq \left|\frac{F_n(\lambda+\epsilon(n))-F_\infty(\lambda+\epsilon(n))}{\epsilon(n)}\right| +  \left|\frac{F_\infty(\lambda)-F_n(\lambda)}{\epsilon(n)}\right| 
 $$
 $$
 +  \left|\frac{F_\infty(\lambda+\epsilon(n))-F_\infty(\lambda)}{\epsilon(n)}-F_\infty^\prime(\lambda)\right| \, .
  $$
 With the special form of  function $f$, the domain $\Delta_f(\lambda)$  is convex, and Theorem \ref{thm1}  provides the estimates
 \begin{align*}
|F_n(\lambda)&-F_\infty(\lambda)| = O\left( 1/n\right)\,, \\
|F_n(\lambda+\epsilon(n))&-F_\infty(\lambda+\epsilon(n))| = O\left(1/n\right)\,,
\end{align*}
where the constant in the $O$-terms   does not depend on $\lambda$ and $\epsilon(n)$. Then, the  first  two terms in Inequality \eqref{ineq} are $O( 1/(n \epsilon(n))$ which tends to $0$ because $n \epsilon(n)\to \infty$.
 For the last term in \eqref{ineq}, we use    Taylor expansion of order 2 of the function $F_\infty$ together Assertion $(a)$.

\section {Conditional expectations. Proof of Thm \ref{thm3}.} 

 We now focus on  conditional expectations. Our final purpose is to prove Theorem \ref{thm3} which is devoted to the recurrence quotient.  However, we begin by a more general study  and we obtain  in Section 4.3 a general result on conditional expectations  (Theorem \ref{conditionalexpectations}). We then apply it in Section  4.4 to the particular case of the recurrence quotient, and this provides Theorem \ref{thm3gen}, which can be viewed  itself as an extension of Theorem \ref{thm3}.

\subsection{Limit expectation of  bounded ${\cal Q}$ functions.} 
Thus far,  we  dealt with distributions of ${\cal Q}$-functions. Now,  we  consider expected values of a ${\cal Q}$-function, and use the equality 
\begin{equation*}
{\mathbb{E} } [\Lambda_n] =\int_0^\infty {\mathbb{P}}[\Lambda_n \geq \lambda] \, d\lambda\,,
\end{equation*}
valid when $\Lambda\geq 0$, as in our case.   We consider  here  the case of a ${\cal Q}$ function  $\Lambda$ associated with a bounded  function $f$
(which is the case when $b_2$ is not zero).  It is then possible to interchange the 
 the limit and the integral and use Theorem \ref{thm1}. \\
 When reversing the order of integration,  we  first integrate  with respect to $\lambda$,  and we are led to  the integral 
 \begin{equation}
\label{eq:expectedomega}
{\mathbb E}_\psi [f] := \frac 6{\pi^2} I(f \cdot 2\omega, {\cal R} ) 
\end{equation}
which is exactly the expectation  ${\mathbb E}_\psi[f]$ of the function $f$  on the rectangle ${\cal R}$ with respect to the density  $\psi :=(12/\pi^2)\,  \omega$.  We  thus obtain the following result   which provides an extension of  Theorem \ref{thm1}:  

\begin{theorem}
\label{thmexpect}
Consider a $\mathcal{Q}$-function $\Lambda$ associated with a function $f$ bounded by $B_f$.
Then 
the sequence $n \mapsto \Lambda_n$   admits a limit expected value  as $n\to\infty$ equal to  the expectation  ${\mathbb E}_\psi[f]$ of the function $f$  on the rectangle ${\cal R}$ with respect to the density  $\psi :=(12/\pi^2)\, \omega$,  and 
\begin{equation}
{\mathbb{E}}\left[\Lambda_n\right] = {\mathbb{E}}_\psi[f]   + B_f \, O\left(  \frac{1}{n}\right)\,,
\end{equation}
where the constant in the $O$-term  does not depend on $f$ and $\lambda$. 
\end{theorem}

\subsection {Case of the recurrence quotient.} 
 The  function $f$ associated with   the recurrence quotient $S(\alpha, n)$ is $f_S(x,y)=1+x+y$. It is  is unbounded on $\mathcal{R}$,  and the function $f_S$  is not integrable with respect to $\psi$.  In fact, by the argument of Proposition \ref{prop:discreteprob} the expected value can be worked out to be
$${\mathbb E}[S_n] = \hat R_n(2\omega f_S),\mathcal{R}),\,$$
and here $ \hat R_n(2\omega f_S,\mathcal{R})$ is infinite for each $n$.

\medskip 
This is why we consider the conditional expectations for the sequence $S_n$ with respect to an event $[\Gamma_n \ge \epsilon(n)]$ associated with another ${\cal Q}$-function $\Gamma$,  namely
$${\mathbb{E}}[S_n|\Gamma_n\geq \epsilon(n)] \, .$$
We will choose in the sequel  the ${\cal Q}$-function $\Gamma$ from the set $\{\mu, \nu, \rho\}$ and a positive sequence $\epsilon(n)$ tending to 0 not all too quickly. 

\subsection{General conditional expectations. }We consider  more general conditional expectations, 
$${\mathbb{E}}[\Lambda_n|\Gamma_n\geq \epsilon]   \qquad (\epsilon >0) $$
when $\Gamma$ is a $\mathcal{Q}$-function associated with  a function $g$ which tends to $0$  for $y \to \infty$. (This means that the  pair  $(b_1, b_2)$ in \eqref{coeff}  satisfies $b_1/b_2 = 0$). The subset 
$$ \{ (x, y) \in {\cal R} \mid g(x, y) \ge \epsilon \} $$
is  bounded for $\epsilon >0$, and we denote, for $\epsilon >0$,  
$$  B_{f|g} (\epsilon) := \sup\{ f(x, y) \mid g(x, y) \ge \epsilon \} < \infty. $$
In this case, the expectation of $f$ with respect to $\psi$ conditioned to the event $[g\ge \epsilon]$ is well defined, and denoted as  $${\mathbb{E}}_\psi[f|g\geq \epsilon]\, .$$
The following holds, and its proof (similar to the proof of  Theorem \ref{thm1}  is in the Annex.

\begin{theorem}
\label{conditionalexpectations}
 Consider two $\mathcal{Q}$-functions $\Lambda$ and $\Gamma$ with  respective associated functions $f$ and $g$.   Assume that $g$ tends to 0 for $y \to \infty$.  Then the conditional expectation of $\Lambda_n$  with respect to the event $[\Gamma_n \ge \epsilon]$ satisfies
\begin{equation*}
{\mathbb{E}}[\Lambda_n|\Gamma_n\geq \epsilon]  \cdot \mathbb{P}[\Gamma_n\geq \epsilon] = {\mathbb{E}}_\psi[f|g\geq \epsilon]  \cdot \mathbb{P}_\psi[g\geq \epsilon] 
\end{equation*}
$$
 	  + B_{f|g} (\epsilon) \, O\left(\frac{1}{n}\right)
$$
where the constant in the $O$-term  does not depend on either $f$, $g$ or $\epsilon$.
\end{theorem}

\subsection{Return to the conditional expectation of the recurrence quotient. Proof of Theorem \ref{thm3}.}

We will prove here a stronger version of Theorem \ref{thm3}, where the remainder terms are more precise. 

\begin{theorem}
\label{thm3gen}
Consider   a parameter $\Gamma \in \{ \rho, \mu, \nu\} $ defined in \eqref{rho} and \eqref{munu}.  
and a sequence $n \mapsto \epsilon (n)$ which tends to 0  with $\epsilon (n) = \Omega(1 /n)$.   
Then the conditional expectation  
of the recurrence quotient  $S_n$  with respect to  the event $[\Gamma_n \ge \epsilon(n)]$ 
satisfies 
$$ 
{\mathbb{E}}\left[  S_n \Big| \Gamma_n\ge   \epsilon(n) \right]   = \frac {12} {\pi^2} |\log \epsilon(n)| +  C(\Gamma)  
$$ 
$$ +  O\left(\frac{1}{\epsilon(n) n} +\epsilon(n)|\log \epsilon(n)|^2\right)\,.$$
Moreover, the constants $C(\Gamma) $ satisfy $$C(\nu ) = +1,\quad  C(\mu) = 0, \quad C(\rho) = +1 \, .$$
\end{theorem}

\begin{proof}
The proof is an application of Theorem \ref{conditionalexpectations}. First, a direct computation with Theorem \ref{thm1} shows that if $\Gamma$ is one of the $\cal Q$-functions $\rho$, $\mu$ or $\nu$,  the following estimates hold
$$
\mathbb{P}[\Gamma_n\geq \epsilon(n)] = 1 + O(\epsilon(n)+1/n) \, .
$$
Along with the bounds and the integrals  provided  in Figure \ref{fig:boundss} (in the Annex),  this implies the result. \EOP
\end{proof}


 Now,  Theorem \ref{thm3} is an immediate application of Theorem \ref{thm3gen} by taking $\epsilon(n)=1/n$.

\section{Comparison between the two models.}

 We now compare the two models, the present model (with a large fixed $n$) and the model which was previously studied in \cite{mfcs},  namely, the model with a large fixed $k$. We first recall our result of \cite{mfcs}, we then  study the number of continuants  $q_k(\alpha)$ in an interval of the form  $[n, cn]$ for $c>1$. 
  
  \subsection{Results in the previous model.} When the integer $n$  of the interval $[q_{k-1}(\alpha), q_k(\alpha)[$  is at a position $\mu$  there, the recurrence  quotient admits the expression
  \begin{equation}\label{smo}
  {S}_{k}^{\langle \mu \rangle}(\alpha) =  f_{\mu}\left(\frac {q_{k-1}(\alpha)}{q_{k}(\alpha)} \right)
  \end{equation} which involves the function
  \begin{equation} 
  f_{\mu}(x):=  1 + \frac  {1 +x} { x+ \mu( 1 - x )}. 
  \end{equation}
  The main idea of the study  ``with a  fixed $k$'' relates the  recurrence quotient and the $k$-th iterate  of the  Euclidean transfer operator ${\bf H}$ via the equality 
  \begin{equation}\label{smo1}
{\mathbb E} [S_{k}^{\langle \mu \rangle}] =  {\bf H}^k \left[x \mapsto \frac {f_{\mu}(x)}{1+x}\right] (0) \, .
   \end{equation}
   The operator ${\bf H}$  admits  nice dominant spectral properties, and, notably,   the celebrated Gauss density   \begin{equation} \label{Gauss}  x \mapsto \left(\frac 1 {\log 2}\right) \,  \frac 1 {1+x}
   \end{equation}  as its fixed density. This leads 
to the estimate
   \begin{equation}\label{smo2}
\lim_{k \to \infty}  {\mathbb E} [S_{k}^{\langle \mu \rangle}]=  1 +  \frac 1 {\log 2} \int_0^1   
\frac 1 {t+ \mu( 1 - t )}
dt \, , 
  \end{equation}
More precisely, we have shown the following in \cite{mfcs}: 
 for the sequence $\mu_k = \tau^k$,  with  $\tau  \in  [\varphi^2, 1[$, (where $\varphi$ is the inverse of the Golden ratio), the following holds  \begin{equation} \label{Eskmuk}
\E[S_k^{\langle \tau^k \rangle}] \sim  \frac 1 {\log 2} \, k  |\log\tau |  \quad   (k \to \infty)\, .
\end{equation}

 \subsection{Relation between the two models.}  
  We now wish  to relate  the two  (asymptotic) models: the present model ``with fixed large $n$'' and the previous model ``with fixed large $k$''?  Of course, these two models  should be close if the  behaviour of the sequence $k \mapsto q_k(\alpha)$  does not depend too strongly on $\alpha$, and we know that it is not the case. However, the behaviour of  the sequence $k \mapsto \log q_k(\alpha)$ is much more regular, as it is well known  (see for instance \cite{khinchin}) that 
  \begin{equation} \label {Levy1}
  \lim_{k \to \infty}  \frac 1 k \log q_k(\alpha) =  L =   \frac{\pi^2}{12 \log 2}  \quad \hbox{for almost all $\alpha$} \, .
  \end{equation}

 Consider   first the present model ``with $n$ fixed'', and a sequence $\ell\mapsto n(\ell)=  \tau^\ell$. Then Theorem \ref{thm3} is  writes  as   
\begin{equation} \label{model2} {\mathbb E} [S_{n(\ell)}  \mid \mu_{n(\ell)} \ge  \tau^{-\ell} ]  \sim \left[ \frac {12}{\pi^2}  \log \tau \right]  \ell\, .\end{equation}
Furhermore, as $ n(\ell) $ belongs to the interval $[q_{k-1}(\alpha), q_k(\alpha)[$, the existence of the limit  for the quotient $q_k(\alpha) /k$, that holds for almost any $\alpha$, and  is recalled in \eqref{Levy1} entails  the  relation between the index $\ell$ and the index $k:= k(\alpha, n(\ell))$,  that holds for almost any $\alpha$, namely
 \begin{equation} \label {Levy}
  \log n(\ell )=  \ell \log \tau  \sim  \frac {\pi^2}{12 \log 2} \, k(\alpha, n(\ell))\,  .
  \end{equation}
Now,  we  deal with the model ``with $k$ fixed'',   and  we consider that the index $k(\alpha, n(\ell))$  satisfies \eqref{Levy} {\em everywhere}. Then,   the application of the result in the model ``with $k$ fixed'', described in \eqref{Eskmuk}  
should entail 
\begin{equation} \label {model1}  \E[S_k^{\langle \tau^k \rangle}] \sim  \left[\frac 1 {\log 2} \log\tau \right] \, k   \sim  \left[\frac {12} {\pi^2}\log^2 \tau \right]\ell \, .  \end{equation}
Remark that the conditional events  are not the same in the two equations \eqref{model2} and \eqref{model1}: 
 
\hskip 0.3cm 
 -- in \eqref{model2}, the  event is $\{ \alpha\mid \mu(\alpha, n(\ell)) \ge \tau^{-\ell}\}$,

\hskip 0.3cm
--  in \eqref{model1} the  event is $\{ \alpha \mid  \mu(\alpha, n(\ell))  \sim \tau^{-\ell}\}$.

This (heuristic)  comparison exhibits in both cases a linear growth with respect to $\ell$. However, the events of interest are not the same, and we have considered  that   the index $k(\alpha, n(\ell))$ satisfies \eqref{Levy} everywhere.

\subsection{Number of continuants in an interval. } 
There is also an interesting connection between the two models, that   counts
the number of terms of the sequence $k \mapsto q_k(\alpha)$ that belongs to the interval $[n, cn[$, for some fixed $c>1$. We thus study the function
$$ (\alpha, n) \mapsto T(\alpha, n) := \sum_{ k \ge 0}  \, [\![ q_k(\alpha) \in [n, cn[ ]\!]\, .$$

 \begin{proposition}{ Consider  the L\'evy constant $\kappa := \exp \left( \pi^2/(12   \log 2) \right)$. Then the  mean number of continuants in the interval $[n , \kappa n]$ tends to 1 as $n \to \infty$}
 \end{proposition}

\begin{proof}
Even if $T$ is not a ${\cal Q}$-function, its expectation ${\mathbb  E}[T_n] $  is expressed as a Riemann sum of the function $2 \omega$, in a domain ${\cal T}_c$. However the domain  ${\cal T}_c$  is not a subset of the rectangle ${\cal R}$. We have indeed
$$ {\mathbb  E}[T_n]  = \int_0^1 T(\alpha, n) d\alpha
 =   \int_0^1  d\alpha \sum_k \, [\![ q_k(\alpha) \in [n, cn[ ]\!]$$
$$= \sum_k\sum_{\m \in {\mathbb N}^k} |I_k(\m)| \,  [\![ q_k(m) \in [n, cn[ ]\!]$$
$$=  \frac 1 {n^2} \sum_k\sum_{\m \in {\mathbb N}^k}   \omega\left( \frac {q_{k-1}(\m)}{n},  \frac {q_{k}(\m)}{n}\right) \left [\!\!\left[ \frac {q_k(\m)}{n} \in [1, c[ \right]\!\!\right]$$
$$ = 2 \sum_{(a, b) \in {\mathbb Z}^2\atop{\gcd(a, b) = 1}}  \omega\left( \frac {a}{n},  \frac {b}{n}\right)  \left[\!\!\left[  \frac {a}{n} \le \frac {b}{n},    1 \le \frac {b}{n} \le c\right]\!\!\right] $$
$$= \hat R_n ( 2\omega, {\cal T}_c),  \quad 
  \hbox{with} \quad  {\cal T}_c =\{ (x, y) \mid x \le y ,  1 \le y \le c\}\, .$$
  Even if ${\cal T}_c$ is not a subset of ${\cal R}$, Proposition \ref{coprimesum} applies, and 
  the coprime Riemann series    admits  a limit equal  to the integral
  $$ \frac {6}{\pi^2} I(2 \omega, {\cal T}_c)= \frac {12}{\pi^2} \frac {\log c}{\log 2} \qquad  \EOP$$
  \end{proof}

\section{Conclusions}

Beginning from the question ``what does the recurrence function of  a  random Sturmian words look like?'', we define and work within a model that is natural at least from an algorithmic standpoint:  pick a  large integer $n$ and let the slope of the word be drawn at random from $[0,1]$.  We are led  to the notion  of the so-called $\cal Q$-functions:  functions that, given $n$ and a slope $\alpha$, place $n$ within the sequence of continuants $k \mapsto q_k(\alpha)$ of $\alpha$, namely consider the index $k$  for which  $n \in [q_{k-1}(\alpha),  q_k(\alpha)[$,  
and then return a value depending only on  the two ratios $(1/n) q_{k-1}(\alpha)$ and  $(1/n) q_{k}(\alpha)$. 
The recurrence quotient  of Sturmian words defines such a ${\cal Q}$ function,  via a Theorem of Morse and Hedlund,  where $n$ is the  length of the factors and $\alpha$ the slope of the word.

\medskip
 Then,  we    study the distribution of  a general $\cal Q$-function. It defines in fact a  sequence of distributions,  and we  prove  that the limit distribution  and the limit densities  exist. They  all involve,   as  a sort of reference density,   the density $\psi$ defined in \eqref{psi}, which plays a similar role to that of the Gauss density  (defined in \eqref{Gauss}) when one studies functions that depend on the ratio $q_{k-1}(\alpha)/q_k(\alpha)$, and appears in our  study  \cite{mfcs}.

 \medskip
 Our results apply in particular to 
  the recurrence quotient of Sturmian words; we  exhibit  the limit distribution (and the limit density) of such a quotient. 
 We wish to compare this probabilistic  study  to the results of  Morse and Hedlund, which exhibit  extreme behaviours, attained  when $n$ is close to the left border $q_{k-1}(\alpha)$ of the interval $[q_{k-1}(\alpha), q_k(\alpha)[$ containing the integer $n$. That is  why we also  consider conditional expectations,  our conditional events  are related to the  various parameters which describe  the  position  of  the integer $n$ inside $[q_{k-1}(\alpha), q_k(\alpha)[$.  We then compare this ``constrained probabilistic'' behaviours to the extreme behaviours, in a precise manner.  
 
 \medskip
We had previously performed a similar study in \cite{mfcs} under another probabilistic model, where it is rather the index $k$ of the interval $[q_{k-1}(\alpha), q_k(\alpha)[$ the integer $n$ belongs to that is fixed. Then for $k \to \infty$, we exhibited limit  distribution and limit densities all of which involve,   as  a sort of reference density,  the Gauss density.  The two models are clearly different,  but the two types of  results show certain similarities.

  \newpage
  
  \section {Annex}
  
  \subsection{Proof of the constant 4 in Proposition 3.2.}
To see where the constant 4  comes from,  we first replace $\Omega(k)$ by a closed convex polygon $\mathcal{C}_n \subset \Omega$, without affecting the  bound :  in each square  $\mathcal{R}_{a,b}$  of the second case, 
pick a point in $\Omega(k)$ and then take the convex hull. If $\Omega(k)$ is a closed convex polygon, we go through the border in clockwise order and look at the grid rectangles we encounter as explained in Figure \ref{fig:convexcount}.  

\begin{figure}[h]
\centering
\includegraphics[width=0.48 \textwidth]{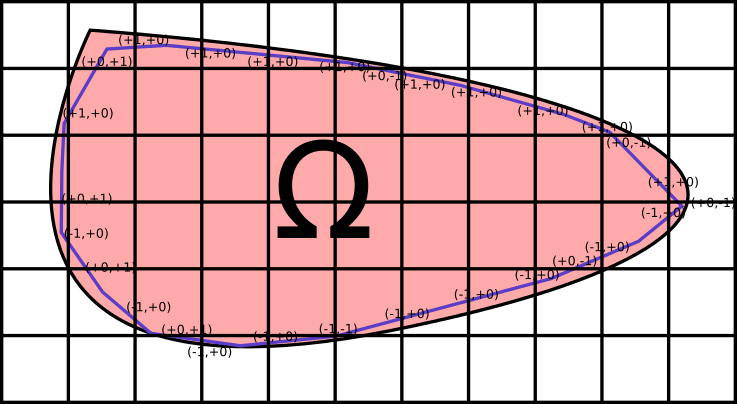} 
\caption{\small The convex domain $\Omega$, and,  in blue,  a convex polytope ${\cal P}$. These two convex sets   have  the  same grid squares that intersect both  themselves and their complement. We traverse the polygon clockwise from the lowest vertex. Each time we intersect a horizontal line we move $\pm 1$ square horizontally in the grid, similarly for the vertical lines, and diagonals. Being the polygon convex, once we stop moving upwards vertically (at most $n$ steps), we can only move downwards (at most $n$ steps) when moving vertically. A similar observation for the horizontal case tells us that there can be at most $2n$ horizontal steps.}
\label{fig:convexcount}
\end{figure}

  \subsection {Proof of  Theorem \ref{thm2} $(a)$.}   
Assertion $(a)$  of Theorem \ref{thm2}  describes the main properties  of the first two derivatives of the function  $\lambda \mapsto I_f(\lambda)$.  These properties are closely related to the geometry  of the figure formed with the rectangle ${\cal R}$ together with   the  set of lines  ${\cal F}$ containing a  given point $(x_0, y_0)$.

\medskip 
{\bf \em The set of lines ${\cal F}$. }
  In the  set ${\cal F}$ of lines, defined as  
 $${\cal F} := \{ f(x, y) = \lambda \mid \lambda \in  {\mathbb R} \}, $$
 the  equation of the line $f(x, y) = \lambda $ is written  in terms of coefficients described in \eqref{coeff} as
 \begin{equation} \label {F}
  (a_1 x + b_1 y + c_1)- \lambda\, ( a_2  x + b_2 y + c_2) = 0  \, .
  \end{equation}
The case where the two vectors $(a_1, b_1, c_1)$ and $(a_2, b_2, c_2)$ are colinear is excluded, as in this case $f(x, y) $ is constant.  The case where $b_1 = b_2 = 0$ is also excluded as we wish that $f$ depend on $y$. Then, there is at most one vertical  line in ${\cal F}$. 

 \smallskip 
There are  two  cases for the  set ${\cal F}$  defined in \eqref{F}.
\begin{itemize} 

\item [$(i)$] the case when the determinant $ r(a, b) := a_1 b_2 - a_2 b_1$ is zero and in this case the determinant $ r(a, c) := a_1 c_2 - a_2 c_1$ is not zero. The set ${\cal F}$ is formed with  parallel lines of slope $-a_1/b_1$. This is for instance  the case of the recurrence quotient  with slope  $-1$  or the case of $\nu$ with slope $0$. 

\item[$(ii)$] the case when the determinant  $ r(a, b) := a_1 b_2 - a_2 b_1$ is not zero.  In this case, we can choose $r(a, b) = 1$ due to the homogeneity of the problem. Then, the set ${\cal F}$ is formed with all the lines which contain the point $(x_0, y_0)$   uniquely defined by 
the relations 
$$ \begin{pmatrix} a_1& b_1\cr
a_2& b_2\cr \end{pmatrix} \begin{pmatrix} x_0\cr y_0\end{pmatrix} \! = \!\begin{pmatrix} -c_1 \cr -c_2\end{pmatrix}  \ \ \hbox{or} \ \  
  \begin{pmatrix} x_0\cr y_0\end{pmatrix} \! =  \! \begin{pmatrix} r(b, c)\cr  -r(a, c)\end{pmatrix} \, .$$
\end{itemize}

\smallskip
The point $(x_0, y_0)$ is called  the basic point of ${\cal F}$. 
Remark that case $(i)$   can be seen as the limit of the case $(ii)$ when $(x_0, y_0 )$ tends to $\infty$ in the direction $a_1/b_1$.    
The basic points attached to  our parameters $\rho, \mu$
 are $ (0, 0) \ $ for $\rho$ and  $(1,1)$ for $\mu$.

 In the set ${\cal F}$   of basic point $(x_0, y_0)$, the value of $\lambda$ and the inverse $1/ \tau$ of  the  slope $\tau$ of the line $f(x, y) = \lambda$,  are related via linear fractional transformations with determinant equal to 1, namely
$$ \lambda =  F(\tau) =  \frac { a_1\tau  + b_1 }{a_2  \tau+ b_2 } \,  \qquad  \tau =  G(\lambda) = \frac {b_2\lambda  -b_1}{ a_2 \lambda -  b_2}\, . $$
  In the set ${\cal F}$   of  basic point $(x_0, y_0)$, the parametrization of the  line  $f(x, y) = \lambda$ of slope $1/\tau$  is  thus 
$$x= x_0 + \tau(y-y_0) , \qquad  \tau = G(\lambda) \, .$$

\medskip 
{\bf \em  Expressions of $I_f$ and its derivative.}
Consider a function $f$ as in \eqref{coeff}; denote by  $\delta_f(\tau)$  the segment  (possibly empty or unbounded) which is the intersection of the line $f(x, y) = \lambda = F(\tau)$ of slope $1/\tau$ with the rectangle ${\cal R}$.  Now, the function $f$ is fixed,  the point $(x_0, y_0)$ is fixed, and all the indices which involve $f$ are removed. 
There is an open interval  $D$ which gathers the values of $\tau$  for which  the segment $\delta(\tau)$ is not empty, and we denote by  $A(\tau), B(\tau)$   the  ordinates of the two ends of the segment $\delta (\tau)$. \\As soon as the line $f(x, y) =  F(\tau) $ is not  horizontal, we consider  the natural  parametrization  $h_\tau$ of the line  $\delta(\tau)$, namely  a map 
 $ h_\tau: ]A(\tau),  B(\tau)[ \rightarrow   \delta(\tau)$  which associates to $y$ the point $$ h_\tau (y) = h(\tau, y) = (x_0 + \tau(y-y_0), y)$$ of the line $\delta(\tau)$. The map $ \tau \mapsto h_\tau$ is of class ${\cal C}^\infty$ on  $D$.  

Using the change of variables $(\theta, y) \mapsto (h(\theta, y), y)$, and its Jacobian $|(\partial h) /(\partial \theta)(y, \theta) | = |y-y_0|$,  the integral $ L(\tau) := L_f (\tau):= I_f(F(\tau)) = I_f \circ F(\tau)$ is written as 
  $$ L(\tau) =  \int_{-\infty }^\tau d\theta \int_{A(\theta)}^{B(\theta)}  Q( \theta, y) dy \, , $$
  $$ \hbox{with} \quad Q(\theta, y) =   \omega(x_0 + \theta(y-y_0), y) \,   |y-y_0| \, .
  $$ 
 (We have used the fact that $F$ is increasing).  Then the derivative of $L$  admits the expression 
   \begin{equation}   \label{Lprime}
  L'(\tau) = \int_{A(\tau)}^{B(\tau)}  Q(\tau, y) dy \, .  
  \end{equation} 
  The function $L'$ is itself differentiable on the set $D$, except  perhaps on a finite set (as we will see now)  and  involves the previous functions under the form
  \begin{equation} \label{Lsecond1}
 L''(\tau) = \int_{A(\tau)}^{B(\tau)} \frac{\partial Q}{\partial \tau}(\tau, y) dy 
 \end{equation}
 \begin{equation} \label{Lsecond2} +  B^\prime(\tau)  \,  Q( \tau, B(\tau)) - A^\prime(\tau)\,  Q( \tau, A(\tau))\,,
 \end{equation}
 $$ \hbox{with} \quad   \frac{\partial Q}{\partial \tau}( \tau, y )=  \frac{\partial \omega}{\partial x} (x_0 + \tau(y-y_0), y) \,  |y-y_0|^2\, .
 $$  We prefer to deal with the function $L$, as it is easier to ``see the geometry''. We will return to  the function $I$ and its two derivatives with the relations
  \begin{equation} \label{der}
   I' (\lambda)  = \frac {L'( \tau)}{F'(\tau)},  \quad  I''(\lambda)  F'(\tau)^2  = L''(\tau) - L'(\tau) \frac {F''(\tau)}{ F'(\tau)} \, , 
   \end{equation}
 and use the special form of $F$ defined in \eqref{F}.
 
  \medskip
{\bf \em  The role of the corners.}   The  values of $\tau$ in  $D$  for which $I'$ is  {\em a priori} not differentiable  are  those for which  the  line  of slope $1/\tau$ is vertical or  meets one of  the two ``corners'' of ${\cal R}$,  namely the slope $1/\tau_0$  for which  it meets   the point $(0, 1)$, and   the slope $1/\tau_1$  for which  it meets   the point $(0, 1)$. \\
There are now  two  different geometric cases:   the generic case $({\cal G})$ or the exceptional case $({\cal E})$, described as follows: 
\begin{itemize}

\item[$({\cal G})$]  If the point $(x_0, y_0)$   does not  belong to the line $y = 1$, there are exactly two lines in ${\cal F}$, each of them containing one corner of $\cal R$,  associated with  two  distinct values $\tau_0$ and $\tau_1$.%

\item[$(
{\cal E})$]  If the point $(x_0, y_0)$    belongs to the line $y = 1$, there is only one  value $\tau_0 = \tau_1 = \infty$. 
\end{itemize}

Finally, there are at most  three values of $\tau$ in the set $\{ 0,  \tau_0, \tau_1\}$ where  $L'$ is possibly not differentiable.   But, $L'$ possesses at each finite $\tau_i$  a left and a right derivative, each of them being finite.  This is thus the same for the derivative  $I'$ of the function $I$. At $\tau = 0$, the  derivatives $F'(0)$ and $F''(0)$  are finite as soon as  $b_2 \not = 0$.  


 \begin{figure*}
\centering
\begin{center}
{\def\arraystretch{2}\tabcolsep=15pt
  \begin{tabular}{| c  || c || c |}
    \hline
    Parameter $\Gamma$ & Bound for $S$ & $    {\mathbb{E}}_\psi[f_S|f_\Gamma\geq \epsilon(n)] \,	\mathbb{P}_\psi[f_\Gamma\geq \epsilon(n)] $  \\ \hline\hline
    $\rho$ & $S \leq 2 + 1/\rho \Longrightarrow B_{f_S|f_\rho}(\epsilon) = O(1/\epsilon)$ & $\ A |\log(\epsilon(n)|  + 1 - A \epsilon(n)|\log \epsilon(n)|$  \\ \hline
    $\mu$ & $S\leq1+1/\mu \Longrightarrow B_{f_S|f_\mu}(\epsilon) = O(1/\epsilon)$ & $\displaystyle\ A |\log\epsilon(n)| +     \frac {A}{1- \epsilon(n)} \epsilon(n)|\log \epsilon(n)| $ \\ \hline
    $\nu$ & $S\leq 1+2/\nu \Longrightarrow B_{f_S|f_\nu}(\epsilon) = O(1/\epsilon)$ & $
   A |\log\epsilon(n)| + 1$ \\
    \hline
  \end{tabular}
  }
\end{center}
        \caption{ In the second column, the bounds for $S$ for each parameter $\Gamma \in \{\rho, \mu, \nu\}$. In the third column,   the values of the product  ${\mathbb{E}}_\psi[f_S|f_\Gamma\geq \epsilon(n)] \,	\mathbb{P}_\psi[f_\Gamma\geq \epsilon(n)] $
        needed to apply Theorem  \ref{conditionalexpectations}. 
        The constant $A$ is  $12/\pi^2$.}
        \label{fig:boundss}
\end{figure*}

\medskip  
{\bf \em  Behaviour of $L'' (\tau)$ for  $\lambda\to 0$.}   
The ratio $ R(\tau) :=B(\tau) /A(\tau)$ is important, as the estimates
 $$ Q(\tau, y) = \Theta (y^{-1}),   \qquad  \frac{\partial Q}{\partial \tau}( \tau, y) =  \Theta (y^{-1})$$
  entail  that $L'(\tau)$ and the first term of $L''(\tau) $ in \eqref{Lsecond1} are  both $\Theta (\log R(\tau))$.

 The bound $B(\tau)$  always tends to $+\infty$ but  there are two cases for  $A(\tau)$: it  remains bounded  or not. 

\begin{itemize}

\item [$(i)$] 
 The case  when   $A(\tau)$ remains bounded  occurs if and only  if the basic point  belongs to one of the two vertical lines $x_0 = 1$ or $x_0 = 1$. Then     the estimates $R(\tau) = \Theta (\tau^{-1})$ and $L'(\tau) = \Theta(\log \tau)$,    directly entail  that $L''(\tau)$ is  $\Theta(\tau ^{-1})$. 

 \item [$(ii)$] 
If $A(\tau)$ tends also to $\infty$,   then the ratio $R(\tau)$ tends to $|x_0 -1| /|x_0|$, and this limit  may be only finite non zero.  Then, the derivatives 
$A'(\tau)$ et $B'(\tau)$ are $\Theta ( \tau ^{-2})$ whereas $A(\tau)$ and $B(\tau)$ are $\Theta (\tau^{-1})$ and thus  $Q( \tau, B(\tau))$ and $Q( \tau, B(\tau))$ are $ \Theta(\tau)$ and  each  term  of \eqref{Lsecond2} is $\Theta (\tau ^{-1})$, whereas the first term  in \eqref{Lsecond1} tends to a finite limit.
More precisely,  the estimate
$$ \tau  \Big( B'(\tau) Q( \tau, B(\tau)) - A'(\tau) Q( \tau, A(\tau)) \Big)   \to 1\, $$
ends with  \eqref{der}  the  proof of Theorem \ref{thm2}  $(a)$.
\end{itemize}

\subsection{Proof of Theorem  \ref{conditionalexpectations}.}
\begin{proof}
The conditional  expectation  is a ratio; the denominator   is ${\mathbb P}[\Gamma_n\geq \epsilon]$ whereas the numerator 
$$ \int_0^\infty {\mathbb{P}}[\Lambda_n \geq \lambda, \Gamma_n \geq \epsilon] \, d\lambda\,.$$
Associate with the pair $(\Lambda, \Gamma)$ its function pair $(f, g)$ and,  for any pair $(\lambda , \epsilon)$ of positive real numbers,  consider the  bounded convex subset  already described in \eqref{Deltafg} 
$$
\und \Delta_{f,g}(\lambda,\epsilon):=\{(x,y)\in \mathcal{R}\mid  f(x,y) \geq \lambda, g(x,y) \geq \epsilon\}\,  .
$$
We have remarked  in Section 3 that  a slight extension of  Proposition \ref{prop:discreteprob}  entails the equality 
$$
{\mathbb{P}}[\Lambda_n\geq \lambda, \Gamma_n \geq \epsilon] = \hat R_n\left( 2 \omega,  \und \Delta_{f,g}(\lambda,\epsilon)\right)\, . 
$$
 Moreover, with the convexity of  the domain  $\und \Delta_{f,g}(\lambda,\epsilon)\subset {\cal R}$,  
Proposition \ref{coprimesum} applies, yielding
$$
{\mathbb{P}}[\Lambda_n\geq \lambda, \Gamma_n \geq \epsilon]  = \frac{12}{\pi^2} I [\omega ,{\und \Delta}_{f,g}(\lambda,\epsilon)] + O\left(\frac{1}{n}\right)\,.
$$
Now we integrate on $\lambda$, noticing that we need only integrate from $0$ to $B_{f|g}(\epsilon)$
$$
\int_0^\infty{\mathbb{P}}[\Lambda_n\geq \lambda, \Gamma_n \geq \epsilon]d\lambda  
$$
 $$
= \frac{12}{\pi^2}\, \int_0^\infty  I [\omega ,{\und \Delta}_{f,g}(\lambda,\epsilon)]  d\lambda+ B_{f|g}(\epsilon)O\left(\frac{1}{n}\right)\,.
$$
 We are led to the integral of $\omega$ on the domain of ${\mathbb R}^3$ defined by 
 $$\{ (x,y,\lambda)\in \cal R \times \mathbb{R}_{\geq 0} \mid  f(x,y)\geq \lambda, g(x,y)\geq \epsilon\} $$
We interchange the summation, and we first integrate with respect to $\lambda$ (which provides the term $f(x, y)$), and obtain
$$\int_0^\infty  I [\omega ,{\und \Delta}_{f,g}(\lambda,\epsilon)]  d\lambda  
= \iint\limits_{(x,y)\in \cal R, \atop{g(x,y)\geq \epsilon}} \omega(x,y) \cdot f(x, y) dxdy\,, 
$$ 
$$
=  \frac{\pi^2}{12}\mathbb{E}_\psi[f|g\geq \epsilon] \cdot \mathbb{P}_\psi[g\geq \epsilon] \, . \qquad \EOP
$$
\end{proof}


\begin{thebibliography}{100}

\bibitem{mfcs}
V.~Berth\'e, E.~Cesaratto, P.~Rotondo, B.~Vall\'ee and A.~Viola,
\emph{Recurrence function of Sturmian Words: a probabilistic study,} in Proceedings of the 40th MFCS, Springer, 2015, Part I, pp. 116--128 


\bibitem{Lot}
M.~Lothaire,
\emph{Algebraic Combinatorics on Words,} 
Encyclopedia of Mathematics and Its Applications~90,
Cambridge University Press, Cambridge, 2002.


\bibitem{Morse}
M.~Morse and G.~Hedlund, \emph{Symbolics dynamics {II}. Sturmian trajectories,}
  American Journal of Mathematics, 62 (1940), pp. 1--42.
  
\bibitem{IoKra}
M.~Iosifescu and C.~Kraaikamp, 
\emph{Metrical Theory of Continued Fractions,} Mathematics and its Applications~547, Kluwer Academic Publishers, Dordrecht, 2002.

\bibitem{khinchin} 
A.~Ya.~Khinchin and H.~Eagle, \newblock {\em Continued Fractions,} \newblock Dover Books on Mathematics.
Dover Publications, Mineola, NY, 1964.

\bibitem{boca} F.~P.~Boca, C.~Cobeli and A.~Zaharescu, \newblock {\em A conjecture of R. R. Hall on Farey points,}  J. Reine Angew. Math., 535 (2001), pp. 207--236.

 \bibitem{hardy} G. H. Hardy and E. M. Wright, \newblock {\em Introduction to the theory of numbers,} \newblock{3rd ed.,  Oxford, at the Clarendon Press, 1954.}
  
  \end{thebibliography}
\end{document}